\documentclass[letterpaper, 10 pt, conference]{ieeeconf}  

\IEEEoverridecommandlockouts                              
\newcommand{\E}[1]{\mathbb{E}\left[ #1 \right]} 

\newcommand{\eps}{\varepsilon} 

\newcommand{\pbr}[1]{\left( #1 \right)} 

\newcommand{\cbr}[1]{\left\{ #1\right\}}

\newcommand{\stepa}[1]{\stackrel{\text{(a)}}{#1}}

\newcommand{\co}{\overline{c}}
\newcommand{\cu}{\underline{c}}
\newcommand{\sco}{\overline{s}_c}
\newcommand{\scu}{\underline{s}_c}

\newcommand{\so}{\overline{s}}
\newcommand{\su}{\underline{s}}
\newcommand{\rh}{\hat{r}}

\usepackage[utf8]{inputenc} 
\usepackage[T1]{fontenc}    
\usepackage[bookmarks, colorlinks=true, plainpages = false, colorlinks=true,
            citecolor=red,
            linkcolor=blue,
            anchorcolor=red,
            urlcolor=blue]{hyperref}
\usepackage{url}

\usepackage{nicefrac}       
\usepackage{microtype}      
\usepackage{xcolor}         

\usepackage{amsmath}
\usepackage{amssymb}
\usepackage{amsfonts}
\usepackage{amsbsy}
\usepackage{bm}
\usepackage{thmtools}
\usepackage{thm-restate}
\usepackage{mathtools}
\usepackage{bbm}
\usepackage{cleveref}
\usepackage{xfrac}

\usepackage{subfigure}
\usepackage{float}
\usepackage{dsfont}
\usepackage{graphicx}
\usepackage{booktabs}
\usepackage{color}
\usepackage{url}

\usepackage[linesnumbered,ruled]{algorithm2e}

\usepackage{hyperref}
\usepackage{amsfonts}
\usepackage{amsmath}
\newtheorem{remark}{Remark}
\newtheorem{proposition}{Proposition}
\newtheorem{theorem}{Theorem}

\title{\LARGE \bf
On the Gaussian Limit of the Output of IIR Filters
}

\author{Yashaswini Murthy, Bassam Bamieh and R. Srikant
\thanks{This work was not supported by any organization}
\thanks{Yashaswini Murthy {\tt\small (ymurthy2@illinois.edu)} and R. Srikant {\tt\small (rsrikant@illinois.edu)} are both with the Electrical and Computer Engineering Department and the Coordinated Science Lab at University of Illinois Urbana-Champaign.
   }%
\thanks{Bassam Bamieh {\tt\small (bamieh@ucsb.edu)} is with the Department of Mechanical Engineering, University of California Santa Barbara.
        }%
}

\begin{document}

\maketitle
\thispagestyle{empty}
\pagestyle{empty}

\begin{abstract}

   
   We study the asymptotic distribution of the output of a stable Linear Time-Invariant (LTI) system driven by a non-Gaussian stochastic input. Motivated by longstanding heuristics in the stochastic describing function method, we rigorously characterize when the output process becomes approximately Gaussian, even when the input is not. Using the Wasserstein-1 distance as a quantitative measure of non-Gaussianity, we derive upper bounds on the distance between the appropriately scaled output and a standard normal distribution. These bounds are obtained via Stein’s method and depend explicitly on the system’s impulse response and the dependence structure of the input process. We show that when the dominant pole of the system approaches the edge of stability and the input satisfies one of the following conditions—(i) independence, (ii) positive correlation with a real and positive dominant pole, or (iii) sufficient correlation decay—the output converges to a standard normal distribution at rate \( O(1/\sqrt{t}) \). We also present counterexamples where convergence fails, thereby motivating thestated assumptions. Our results provide a rigorous foundation for the widespread observation that outputs of low-pass LTI systems tend to be approximately Gaussian.

\end{abstract}

\section{INTRODUCTION}

A large class of nonlinear time-invariant systems 
can be modeled as shown in Figure~\ref{LTI_fdbk_N.fig} as a Linear Time-Invariant (LTI) system in feedback 
with a memoryless nonlinearity $N(.)$. The well-known describing function method~\cite{Atherton,vander1968multiple} 
predicts periodic oscillation solutions based on the assumption that the LTI part has a ``low-pass character'', and that 
higher harmonics in the output of $N(.)$ can be neglected. Parameters of periodic (limit cycle) solutions can then be 
obtained from algebraic equations involving the frequency response of the LTI part together with the describing 
function of $N(.)$.

A stochastic analog of this method concerns computing the 2nd-order statistics of the output process $z$ given those 
of the input process $d$ and the dynamics. The method as outlined in~\cite[Chapter 7]{vander1968multiple} can be 
roughly described as follows. If  the LTI part ($H$ in Figure~\ref{LTI_fdbk_N.fig}) has a 
``sufficiently low-pass character'', it has long been observed that the output processes $z$ and $v$ are ``nearly Gaussian''. 
Assuming $v$ to be Gaussian then enables computing the 2nd-order statistics of the output $w$ of the nonlinearity $N(.)$. This 
can be done for example with Hermite polynomials when $N(.)$ is itself a polynomial nonlinearity~\cite[Appendix H]{vander1968multiple}. 
In turn, given the 2nd-order statistics of $d$ and $w$, e.g. the Power Spectral Densities (PSDs) (or the autocorrelation functions)
 of the outputs $z$ and $v$ can be 
computed from the frequency response (or impulse response) of $H$. This gives a set of algebraic equations for the 2nd-order 
statistics of $w,v,z$ given those of $d$. The equations can then be solved by several techniques such as fixed-point iterations.

\begin{figure}[t]
	\centering
	\includegraphics[width=0.3\textwidth]{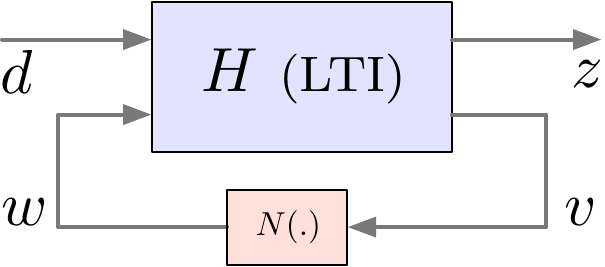} 
	
	\caption{A large class of nonlinear time-invariant dynamic systems (from $d$ to $z$) can be modeled as a Linear Time-Invariant 
		(LTI) system in feedback with a memoryless nonlinearity $N(.)$. 
		} 
  \label{LTI_fdbk_N.fig}
\end{figure} 

The key to the stochastic version of the method is the informal observation that outputs of low-pass type LTI systems are nearly 
Gaussian even if the input processes are not.  However, as stated in~\cite[pg. 641]{vander1968multiple} ``Unfortunately, it does
not seem possible to phrase this general statement in quantitative terms.'' Despite the passage of time, this issue has not been 
quantitatively addressed in the literature to the author's best knowledge. 
The question is of some importance since numerical experiments indicate 
that the success and convergence of the aforementioned fixed-point iterations appear to depend heavily on how close 
the output processes are to being Gaussian. 

In this paper, we address this question directly using the Wasserstein-1 distance as a quantitative measure of how far the output 
processes are from being Gaussian. Specifically, we
 consider a stable linear system of the form $y_t = \sum_{i=1}^t G_{t-i}u_i$, where the input $\{u_i\}$ is a noise sequence and $y_t$ is the output of the linear system. We show that, when appropriately scaled, \( y_t \) converges in distribution to a standard normal random variable $Z$ with an asymptotic convergence rate of \( O(1/\sqrt{t}) \), provided the dominant pole of the system response \( G \) approaches 1, under the following conditions:
\begin{itemize}
    \item[(i)] The sequence \( \{u_i\} \) is independent;
    \item[(ii)] The \( u_i \) are positively correlated, and the system \( G \) has a dominant pole that is real and positive;
    \item[(iii)] The sequence \( \{u_i\} \) exhibits sufficient correlation decay;
\end{itemize}
Conversely, we identify specific examples of input sequences \( \{u_i\} \) for which \( y_t / \sigma_t \not\!\xrightarrow{d} Z \), thereby demonstrating that convergence to a standard normal distribution is not guaranteed for arbitrary inputs.

The key idea is to use the standard upper bound on the Wasserstein-1 distance between the scaled output and the standard Gaussian random variable, derived via Stein's method \cite{barbour2005introduction,chen2010normal,chatterjee2014shortsurveysteinsmethod}. We then show that the upper bound can be explicitly characterized in terms of the dominant root of the system's transfer function.  Finally, by analyzing the asymptotic behavior of the resulting expression, we establish the desired convergence rate.

The remainder of the paper is organized as follows. In \Cref{sec:Model}, we introduce the system model under consideration. \Cref{section:mainresults} presents the main theorem statements. Selected proofs are provided in \Cref{sec:Proofs}, while the remaining proofs are deferred to the Appendix. The final section concludes with a summary of results and potential directions for future work.

\section{MODEL AND PRELIMINARIES}
\label{sec:Model}
Consider a linear time-invariant (LTI) system \( G \) with input signal \( \{u_t\}_{t \geq 1} \subset \mathbb{R} \) and scalar output signal \( \{y_t\}_{t \geq 1} \subset \mathbb{R} \). The output \( y_t \) is given by the convolution of the system’s impulse response \( \{G_t\}_{t \geq 0} \) with the input:
\begin{equation}
   \label{eq:finalsystem_mb}
   y_t = \sum_{i=1}^t G_{t-i} u_i, \quad \forall t \geq 1.
\end{equation}
For motivation behind such a formulation, refer to \Cref{subsection:motivation}.

Assume the input signal is zero-mean, i.e., \( \mathbb{E}[u_i] = 0 \) for all \( i \geq 0 \). Then, by linearity, \( \mathbb{E}[y_t] = 0 \) for all \( t \geq 0 \). Consequently, the variance of \( y_t \) is given by
\begin{equation}
   \label{eq:variance}
   \sigma_t^2 := \text{Var}(y_t) = \mathbb{E} \left[ \left( \sum_{i=1}^t G_{t-i} u_i \right)^2 \right].
\end{equation}


\begin{center}
   \fbox{%
     \begin{minipage}{0.95\linewidth}
      Objective: To determine conditions under which  
   \begin{equation}
   \frac{y_t}{\sigma_t} \xrightarrow{d} \mathcal{N}(0,1),
   \end{equation}  
   and, when convergence in distribution holds, to characterize the rate at which this convergence occurs.
     \end{minipage}%
   }%
\end{center}
\vspace{1mm}

To quantify the rate of convergence to the normal distribution, we use the {Wasserstein-1 distance}. For random variables \(X\) and \(W\), the Wasserstein-1 distance \( d_{\mathcal{W}}(X, W) \) is defined as
\begin{equation}
   \label{eq:W_defn}
   d_{\mathcal{W}}(X, W) := \sup_{h \in \text{Lip}_1} \mathbb{E}[h(X) - h(W)],
\end{equation}
where \( \text{Lip}_1 \) denotes the set of all 1-Lipschitz functions. Notably, convergence in the Wasserstein-1 metric implies convergence in distribution. 

Stein~\cite{stein1972bound} established that the Wasserstein-1 distance between a random variable \( X \) and the standard normal variable \( Z \sim \mathcal{N}(0,1) \) can be bounded as follows:
\begin{equation}
   \label{eq:W_stdN}
   d_{\mathcal{W}}(X, Z) \leq \sup_{f \in \mathcal{F}} \mathbb{E}[f'(X) - Xf(X)],
\end{equation}
where \( \mathcal{F} \) denotes a class of functions with uniformly bounded first and second derivatives.

Let \( \{u_i\} \) be a sequence of random variables with a local dependence structure, meaning that each \( u_i \) is dependent on at most \( D \) other variables \( u_j \). Such a dependence formulation arises naturally if \eqref{eq:finalsystem_mb} is generated by an ARMA process (see \Cref{subsection:motivation} for more details.) The parameter \( D \) is referred to as the maximum dependency degree of the sequence \( \{u_i\} \). Then the following proposition holds:

\begin{proposition}
\label{Prop1}
Consider the system defined in \eqref{eq:finalsystem_mb}. Suppose the input sequence \( \{u_i\} \) satisfies the following conditions:
\begin{enumerate}
    \item Each \( u_i \) has maximum dependency degree \( D \).
    \item \( \mathbb{E}[|u_i|^3] \leq s_3 \) and \( \mathbb{E}[|u_i|^4] \leq s_4 \), for some finite positive constants \( s_3, s_4 > 0 \).
\end{enumerate}
Then, for all \( t \geq 1 \), the Wasserstein-1 distance between the normalized output \( y_t / \sigma_t \) and the standard normal random variable \( Z \sim \mathcal{N}(0,1) \) satisfies:
\begin{align}
   d_{\mathcal{W}}\left( \frac{y_t}{\sigma_t}, Z \right) 
   &\leq \frac{D^2}{\sigma_t^3} \sum_{i=1}^t \mathbb{E}[|G_{t-i} u_i|^3] \nonumber \\
   &\quad + \frac{2D^{3/2}}{\sqrt{\pi} \sigma_t^2} \sqrt{ \sum_{i=1}^t \mathbb{E}[|G_{t-i} u_i|^4] },
   \label{eq:dwbound}
\end{align}
where \( \sigma_t^2 = \operatorname{Var}(y_t) \) is defined in \eqref{eq:variance}.
\end{proposition}
\begin{proof}
   See  \cite[Section 3.1.1]{Zhang2016FUNDAMENTALSOS}.
\end{proof}
The proof of this proposition can also be found in \cite{ross2011fundamentals}. However, we borrow the result from~\cite{Zhang2016FUNDAMENTALSOS}, as it provides tighter constants.

Our goal is to understand the conditions under which the right-hand side of (\ref{eq:dwbound}) goes to zero when $t\rightarrow 0$ and, under these conditions, characterize the rate of convergence to zero.

\section{MAIN RESULTS}
\label{section:mainresults}

Before we state our main results, we provide some definitions. Let \( \mathcal{N}_{u_i} \) denote the set of all \( \{u_j\} \) such that \( u_j \) is correlated with \( u_i \). The sequence \( \{u_i\} \) is said to be \emph{positively correlated} if $\forall\ u_i,$
\begin{equation}
   \label{eq:pos_cor}
   \mathbb{E}[u_i u_j] > 0 \quad  \ \forall\, u_j \in \mathcal{N}_{u_i}.
\end{equation}

The sequence $\{u_i\}$ is said to satisfy a correlation decay property if
\begin{equation}
   |\E{u_i u_j}| \leq a \underbar{s}_2 \ \ \forall i\neq j
\end{equation}
where $0<a<1$ is some constant and $\underbar{s}_2:= \min_i \E{u_i^2}$. 
\begin{theorem}
   \label{thm:main}
   Consider the LTI system described in \eqref{eq:finalsystem_mb} with output \( y_t \) and variance \( \sigma_t^2 = \operatorname{Var}(y_t) \). Let \( Z \sim \mathcal{N}(0,1) \) denote the standard normal distribution. 
   Suppose any of the following is true: 
      \begin{itemize}
         \item[(i)] \( \{u_i\} \) is an independent sequence;
         \item[(ii)] \( \{u_i\} \) is positively correlated, and the system \( G \) has a dominant pole that is real and positive;
         \item[(iii)] \( \{u_i\} \) exhibits correlation decay for $a$ sufficiently small, where the upper bound on $a$ depends on the parameters of the LTI system (see \Cref{subsection:GCD}).
     \end{itemize}
      Then,
      \begin{equation}
         d_{\mathcal{W}}\left(\frac{y_t}{\sigma_t},Z\right) \leq f(\alpha,t)
      \end{equation}
      where \( e^{-\alpha} := |r| \), with \( r \) denoting the dominant pole of the system response \( G \). The exact form of the function \( f(\alpha, t) \) for each case above is provided in the corresponding proofs.

      Furthermore, the asymptotic rate of convergence as the dominant pole approaches the edge of stability satisfies:
      \begin{equation}
         \lim_{\alpha \to 0} d_{\mathcal{W}}\left( \frac{y_t}{\sigma_t}, Z \right) \leq O\left( \frac{1}{\sqrt{t}} \right).
      \end{equation}
   \end{theorem}
   \vspace{3mm}
   \begin{proof}
      The proofs corresponding to (i), (ii), and (iii) are presented in \Cref{subs:independent}, \Cref{subs:pos_cor}, and \Cref{subsection:GCD}, respectively.
   \end{proof}

Theorem 1 identifies sufficient conditions under which the output of a stable LTI system is asymptotically Gaussian. While do we do not have a complete characterization in terms of matching necessary conditions, in the appendix we provide two counterexamples which illustrate that the Gaussian limit may not hold in general when the above conditions are violated.
   
\section{PROOFS}
\label{sec:Proofs}

\subsection{\( \{u_i\} \) is an independent sequence}
\label{subs:independent}
Recall the definition of \( \sigma_t \) in \eqref{eq:variance}. Then, \( \sigma_t^2 \) is given by:
\[
\!\sigma_t^2 \!= \!\mathbb{E} \Big[ \Big( \sum_{i=1}^t G_{t-i} u_i \Big)^2 \Big] 
\stepa{=} \mathbb{E} \Big[ \sum_{i=1}^t G_{t-i}^2 u_i^2 \Big] 
\!=\! \sum_{i=1}^t G_{t-i}^2\, \mathbb{E}[u_i^2],
\]
where step (a) uses \( \mathbb{E}[u_i] = 0 \) and \( \mathbb{E}[u_i u_j] = \mathbb{E}[u_i] \mathbb{E}[u_j] = 0 \) for \( i \neq j \).

Assuming \( \mathbb{E}[u_i^2] \geq s_2 \) for all \( i \), it follows that
\[
\sigma_t^2 \geq s_2 \sum_{i=1}^t G_{t-i}^2.
\]

Since \( \{u_i\} \) is an independent sequence, the dependence degree is \( D = 1 \) (each \( u_i \) depends only on itself). Then from \Cref{Prop1}, the Wasserstein-1 distance between \( y_t/\sigma_t \) and the standard normal \( Z \sim \mathcal{N}(0,1) \) satisfies
\[
d_{\mathcal{W}}\left( \frac{y_t}{\sigma_t}, Z \right) \leq A + B,
\]
where
\[
A := \frac{\sum_{i=1}^t |G_{t-i}|^3 s_3}{s_2^{3/2} \left( \sum_{i=1}^t G_{t-i}^2 \right)^{3/2}}, \quad
B := \frac{2 \sqrt{ s_4 \sum_{i=1}^t |G_{t-i}|^4 }}{\sqrt{\pi} \cdot s_2 \left( \sum_{i=1}^t G_{t-i}^2 \right)}.
\]
Here, \( s_3 \) and \( s_4 \) are as defined in \Cref{Prop1}.

Assuming that the system in \eqref{eq:finalsystem_mb} arises from an underlying ARMA process, the following bounds hold (see \Cref{subsection:System_Response} for details). There exist constants \( \cu_{d_1}, \co_{d_1}, \epsilon, \alpha > 0 \) and \( d \in \mathbb{W} \) such that for all \( i \geq T_\epsilon \), where \( T_\epsilon \in \mathbb{N} \) is a model-dependent constant, we have:
\begin{align}
   \label{eq:bounds_on_G}
   |G_i| &\leq \co_{d_1} \cdot e^{\epsilon} \cdot e^{d \log(i) - \alpha i}, \\
   |G_i| &\geq \cu_{d_1} \cdot e^{-\epsilon} \cdot e^{d \log(i) - \alpha i}.
\end{align}
Here, \( e^{-\alpha} := |r| \), where \( r \) is the dominant pole of the system response \( G \), as in \Cref{thm:main}.

This allows us to bound the term $A$ as
\begin{align}
    A \!&= \frac{s_3 \pbr{\sum_{i=1}^{T_{\eps}-1} |G_i|^3 + \sum_{i=T_{\eps}}^t |G_i|^3} }{ s_2^{3/2} \pbr{\sum_{i=1}^{T_{\eps}-1}|G_i|^2 + \sum_{i=T_{\eps}}^t |G_i|^2}^{3/2}} \nonumber \\
    &= \frac{s_3 \pbr{ \co_3(T_\eps) + \sum_{i=T_{\eps}}^t |G_i|^3 }}{ s_2^{3/2} \pbr{\cu_2(T_\eps) + \sum_{i=T_\eps}^{t}|G_i|^2}^{3/2}} \nonumber \\
    & \leq \frac{ s_3 \pbr{  \co_3(T_\eps) \!+ \co_{d_1}^3  e^{3\eps}  \sum_{i=T_\eps}^{t} e^{3d\log i  - 3\alpha i}  }}{ s_2^{3/2} \pbr{ \cu_2(T_\eps) \!+ \cu_{d_1}^2  e^{-2\eps}  \sum_{i=T_\eps}^{t} e^{ 2d\log i - 2\alpha i} }}. \label{eq: bound-on-A}
\end{align}
Here, we consider two subcases.

\noindent \textit{Case A.1: $d \geq 1$.} Setting $\alpha = 0$ (which is equivalent to being at the edge of stability) in~\eqref{eq: bound-on-A} yields,
\begin{align*}
    A \leq \frac{ s_3 \pbr{  \co_3(T_\eps) + \co_{d_1}^3  e^{3\eps}  \sum_{i=T_\eps}^{t} i^{3 d}  }}{ s_2^{3/2} \pbr{ \cu_2(T_\eps) + \cu_{d_1}^2  e^{-2\eps}  \sum_{i=T_\eps}^{t} i^{2 d} }}.  
\end{align*}
It follows that there exist constants $c', c''$ such that 
\begin{align}
    A &\leq \frac{c' \sum_{i=T_\eps}^t i^{3d}}{c'' \pbr{\sum_{i=T_\eps}^t i^{2d}}^{3/2}} \leq \frac{c' t^{d} \cdot \sum_{i=T_\eps}^t i^{2d}}{ c'' \pbr{\sum_{i=T_\eps}^t i^{2d}}^{3/2}} \nonumber \\
    & = \frac{c'}{c''} \cdot \sqrt{ \frac{t^{2d}}{ \sum_{i=T_{\eps}}^t i^{2d}}} \leq \frac{c'}{c''} \cdot \sqrt{\frac{t^{2d}}{ \sum_{i=t/2}^t i^{2d}}} \nonumber \\
    & \leq \frac{c'}{c''} \cdot \sqrt{ \frac{t^{2d}}{ (t/2) \cdot (t/2)^{2d}}} = O(1/\sqrt{t}). \label{eq: A-is-O(1/sqrt(t))}
\end{align}

\noindent \textit{Case A.2: $d=0$.} Here we have
\begin{align*}
    A &\leq \frac{s_3 \pbr{ \co_3(T_\eps) + e^{3\eps} \co_{d_1}^3 \sum_{i=T_\eps}^t e^{-3\alpha i}}}{s_2^{3/2} \pbr{\cu_2(T_\eps) + e^{-2\eps} \cu_{d_1}^2 \sum_{i=T_\eps}^t e^{-2\alpha i}}^{3/2}}  \\
    & \leq \frac{s_3 \pbr{\co_3(T_\eps) + e^{3\eps} \co_{d_1}^3 \sum_{i=T_\eps}^t (1 - 3\alpha i + 9\alpha^2i^2/2)} }{ s_2^{3/2} \pbr{\cu_2(T_\eps) + e^{-2\eps} \cu_{d_1}^2 \sum_{i=T_\eps}^t (1-2\alpha i)}^{3/2}},
\end{align*}
where the last step uses that $1-x \leq e^{-x} \leq 1- x + x^2/2$ for all $x$. It follows that there exist positive constants $p'$ and $p''$ such that
\begin{align}
    A \leq \frac{p' \cdot \sum_{i=T_\eps}^t (1-3\alpha i + 9\alpha^2i^2/2)}{p'' \pbr{\sum_{i=T_\eps}^t (1-2\alpha i)}^{3/2}}. \label{eq: alpha-dependent-bound-on-A}
\end{align}
Setting $\alpha = 0$ in~\eqref{eq: alpha-dependent-bound-on-A} yields that
\begin{equation}
   \label{eq:A}
   A \leq \frac{p' \cdot (t-T_\eps)}{p'' \cdot (t-T_\eps)^{3/2}} = O(1/\sqrt{t}).
\end{equation}

Next, the term $B$ is analyzed. Notice that
\begin{align}
   \label{eq:initialB}
    B \leq \frac{ 2 \sqrt{ s_4 \pbr{\co_4(T_\eps) + \sum_{i=T_\eps}^t \co_{d_1}^4 e^{4\eps} e^{4(d\log i - \alpha i)}}} }{\sqrt{\pi} s_2 \pbr{ \cu_2(T_\eps) + \sum_{i=T_\eps}^t \cu_{d_1}^2 e^{-2\eps} e^{2(d\log i - \alpha i)} }}.
\end{align}

\noindent \textit{Case B.1: $d\geq 1$.} Setting $\alpha = 0$ yields
\begin{align*}
    B \leq \frac{2 \cdot \sqrt{s_4 \pbr{\co_3(T_\eps) + \co_{d_1}^4 e^{4\eps} \sum_{i=T_{\eps}}^t i^{4d}}}}{s_2^{3/2} \pbr{\cu_2(T_\eps) + \cu_{d_1}^2 e^{-2\eps}\sum_{i=T_\eps}^t i^{2d}}}.
\end{align*}
It follows that there exist constants $c_3'$ and $c_3''$ such that 
\begin{align}
    B \!\leq \frac{c_3' \sqrt{\sum_{i=T_\eps}^t i^{4d}}}{c_3'' \sum_{i=T_{\eps}}^t i^{2d}} \! \leq \frac{c_3'}{c_3''} \sqrt{\frac{t^{2d}}{\sum_{i=T_\eps}^t i^{2d}}} \stepa{=} O(1/\sqrt{t}), \label{eq: B-is-O(1/sqrt(t))}
\end{align}
where (a) uses the same steps as in~\eqref{eq: A-is-O(1/sqrt(t))}.

\noindent \textit{Case B.2: $d=0$.} Here we have
\begin{align*}
    B &\leq \frac{2 \sqrt{ s_4 \pbr{\co_4(T_\eps) + \sum_{i=T_\eps}^t \co_{d_1}^4 e^{4\eps} e^{-4\alpha i}} }}{\sqrt{\pi} s_2 \pbr{\cu_2(T_\eps) + \sum_{i=T_\eps}^t \cu_{d_1}^2 e^{-2\eps} e^{-2\alpha i}}} \\
    & \leq \frac{2 \sqrt{ s_4 \pbr{\co_4(T_\eps) +  \co_{d_1}^4 e^{4\eps}\sum_{i=T_\eps}^t (1-4\alpha i + 8\alpha^2 i ^2) } }}{\sqrt{\pi} s_2 \pbr{\cu_2(T_\eps) +  \cu_{d_1}^2 e^{-2\eps} \sum_{i=T_\eps}^t (1-2\alpha i)}},
\end{align*}
where the last step uses that $1-x\leq e^{-x} \leq 1-x+x^2/2$ for all $x$. It follows that there exist positive constants $p_1'$ and $p_1''$ such that
\begin{align}
    B \leq \frac{p_1' \sqrt{\sum_{i=T_\eps}^t (1-4\alpha i + 8 \alpha^2 i^2)}}{p_1''  \cdot \sum_{i=T_\eps}^t (1-2\alpha i)} \label{eq: alpha-dependent-bound-on-B}.
\end{align}
Setting $\alpha=0$ in~\eqref{eq: alpha-dependent-bound-on-B} yields that
\begin{equation}
   \label{eq:B}
B \leq \frac{p_1' \cdot \sqrt{t}}{p_1''\cdot  t} = O(1/\sqrt{t}).
\end{equation}

Note that the function \( f(\alpha, t) \) in \Cref{thm:main}, corresponding to the case where \( \{u_i\} \) is an independent sequence, is obtained by summing the right-hand sides of \eqref{eq: bound-on-A} and \eqref{eq:initialB}.

The asymptotic rate of convergence as the dominant pole approaches the edge of stability follows from \eqref{eq: A-is-O(1/sqrt(t))}, \eqref{eq:A}, \eqref{eq: B-is-O(1/sqrt(t))} and \eqref{eq:B}. Specifically, we obtain
\[
\lim_{\alpha\to 0}d_{\mathcal{W}}\left( \frac{y_t}{\sigma_t}, Z \right) \leq O\left(\frac{1}{\sqrt{t}}\right).
\]

\subsection{\( \{u_i\} \) is positively correlated}
\label{subs:pos_cor}
Given a $u_i$, let $\mathcal{N}_{u_i}$ represent the set of inputs that are correlated with $u_i$. For more details on the structure of  $\mathcal{N}_{u_i}$, refer \Cref{subsection:motivation}. 
The variance of $y_t$, denoted by $\sigma_t^2$ in \eqref{eq:variance} is given by:
\begin{align}
   \sigma_t^2 & = \E{\sum_{i=1}^t\sum_{j=1}^t G_{t-i}G_{t-j}u_i u_j} \\ & = \E{\sum_{i=1}^t\sum_{j\in\mathcal{N}_{u_i}} G_{t-i}G_{t-j}u_i u_j} 
\end{align}
Upon rearranging the indices, 
\begin{align}
   \sigma_t^2 & =  \E{\sum_{i=1}^t\sum_{t-j\in\mathcal{N}_{u_{t-i}}} u_{t-i}u_{t-j}G_i G_j} \\
   & = \E{\sum_{i=1}^{T'_\eps-1}\sum_{t-j\in\mathcal{N}_{u_{t-i}}}\!\!\!\!\!\!u_{t-i}u_{t-j}G_i G_j} \!\!\! \\ & ~~~~~~~ + \E{\sum_{i=T'_\eps}^t\sum_{t-j\in\mathcal{N}_{u_{t-i}}}\!\!\!\!\!\!u_{t-i}u_{t-j}G_i G_j} 
\end{align}
Let $\scu \leq \E{u_i u_j} \leq \sco \ \forall i, j\in\mathcal{N}_{u_i}$. Since $\{u_i\}$ are positively correlated, it follows from \eqref{eq:pos_cor} that $\scu>0$ when $u_j$ is correlated to $u_i$. There exists a finite \( T'_\epsilon \in \mathbb{N} \) such that for all \( i, j > T'_\epsilon \), we have \( G_i G_j = |G_i| |G_j|.\) Further details can be found in \Cref{subsubsection:pos_cor}. We then obtain,
\begin{align}
   \sigma_t^2 & \geq c(T'_\eps) + \scu\left[\sum_{i=T'_\eps}^t\sum_{t-j\in\mathcal{N}_{u_{t-i}}}|G_i||G_j|\right]
\end{align}
From \eqref{eq:bounds_on_G},
\begin{align}
   \sigma_t^2 & \geq c(T'_\eps) + \scu e^{-2\eps}\cu_d^2 \left[\sum_{i=T'_\eps}^t\sum_{t-j\in\mathcal{N}_{u_{t-i}}} i^d j^d e^{-\alpha i}e^{-\alpha j}\right]
\end{align}
Let \( M \in \mathbb{N} \) be such that either \( i - M \) is the smallest index for which \( u_{i-M} \in \mathcal{N}_{u_i} \), or \( i + M \) is the largest index for which \( u_{i+M} \in \mathcal{N}_{u_i} \). Then,
\begin{align}
   \sigma_t^2 & \geq c(T'_\eps) + \scu e^{-2\eps}\cu_d^2 \left[\sum_{i=T'_\eps}^t i^d (i-M)^d e^{-\alpha i}e^{-\alpha (i+M)}\right] \\
   & = c(T'_\eps) + \scu e^{-2\eps-M\alpha}\cu_d^2 \left[\sum_{i=T'_\eps}^t i^d (i-M)^d e^{-2\alpha i}\right] \label{eq:boundonsigma}
\end{align}
From \Cref{Prop1}, the Wasserstein-1 distance between \( y_t/\sigma_t \) and the standard normal \( Z \sim \mathcal{N}(0,1) \) satisfies
\[
d_{\mathcal{W}}\left( \frac{y_t}{\sigma_t}, Z \right) \leq P + Q,
\]
where
\[
P := \frac{D^2}{\sigma_t^3} s_3 \sum_{i=1}^t |G_{t-i}|^3 , \quad
Q := \frac{2 D^{3/2}}{{\sqrt{\pi} \sigma_t^2}}\sqrt{ s_4 \sum_{i=1}^t |G_{t-i}|^4 }.
\]
Here, \( s_3 \) and \( s_4 \) are as defined in \Cref{Prop1}.

The term \( P \) can initially be bounded in a manner similar to the numerator of \eqref{eq: bound-on-A}, as shown below:
\begin{align}
   P \leq \frac{D^2}{\sigma_t^3} s_3 \pbr{  \co_3(T'_\eps) \!+ \co_{d_1}^3  e^{3\eps}  \sum_{i=T'_\eps}^{t} e^{3d\log i  - 3\alpha i}  }
\end{align}
From \eqref{eq:boundonsigma}, 
\begin{equation}
   \label{eq: bound_on_P}
   P\leq \frac{D^2 s_3 \pbr{  \co_3(T'_\eps) \!+ \co_{d_1}^3  e^{3\eps}  \sum_{i=T'_\eps}^{t} e^{3d\log i  - 3\alpha i}  } }{\pbr{c(T'_\eps) + \scu e^{-2\eps-M\alpha}\cu_d^2 \left[\sum_{i=T'_\eps}^t i^d (i-M)^d e^{-2\alpha i}\right]}^{3/2}}
\end{equation}
There exists some constants $q',q''$ such that
\begin{equation}
   P \leq \frac{q'\sum_{i=T'_\eps}^t e^{-3\alpha i}i^{3d}}{q''\pbr{\sum_{i=T'_\eps}^t e^{-2\alpha i}i^{2d}}^{3/2}}
\end{equation}
Further analysis of bounding of $P$ is similar to \textit{Case A.1} and \textit{Case A.2} in \Cref{subs:independent}. 

The term \( Q \) can initially be bounded in a manner similar to the numerator of \eqref{eq:initialB}, as shown below:
\begin{equation}
   Q \leq \frac{2 D^{3/2}}{{\sqrt{\pi} \sigma_t^2}} \sqrt{ s_4 \pbr{\co_4(T'_\eps) + \sum_{i=T'_\eps}^t \co_{d_1}^4 e^{4\eps} e^{4(d\log i - \alpha i)}}}
\end{equation}
From \eqref{eq:boundonsigma},  
\begin{equation}
   \label{eq: bound_on_Q}
   Q \leq \frac{2 D^{3/2}  \sqrt{ s_4 \pbr{\co_4(T'_\eps) + \sum_{i=T'_\eps}^t \co_{d_1}^4 e^{4\eps} e^{4(d\log i - \alpha i)}}}}{\sqrt{\pi}c(T'_\eps) + \scu e^{-2\eps-M\alpha}\cu_d^2 \left[\sum_{i=T'_\eps}^t i^d (i-M)^d e^{-2\alpha i}\right] }
\end{equation}
There exists some constants $l',l''$ such that
\begin{equation}
   Q\leq \frac{l'\sqrt{\sum_{i=T'_\eps}^t e^{-4\alpha i}i^{4d}}}{l''\pbr{\sum_{i=T'_\eps}^t e^{-2\alpha i}i^{2d}}}
\end{equation}
Further analysis of bounding of $Q$ is similar to \textit{Case B.1} and \textit{Case B.2} in \Cref{subs:independent}. 

Note that the function \( f(\alpha, t) \) in \Cref{thm:main}, corresponding to the case where \( \{u_i\} \) is positively correlated, is obtained by summing the right-hand sides of \eqref{eq: bound_on_P} and \eqref{eq: bound_on_Q}.

Hence, the asymptotic rate of convergence as the real, positive dominant pole approaches the edge of stability is,
\[
   \lim_{\alpha\to 0} d_{\mathcal{W}}\left( \frac{y_t}{\sigma_t}, Z \right) \leq O\pbr{\frac{1}{\sqrt{t}}}.
\]

\section{Counterexamples}

\subsubsection{When $u_i$ are postiviely correlated and the dominant pole is in the left half plane}
Consider the following example,
\[ 
y_t = - y_{t-1} + w_t +w_{t-1},
\]
where $w_i$ are independent. Equivalently, 
\[ 
y_n = w_n + (-1)^{n+1}w_0.
\]
Such a system does not converge to a standard normal. Let $u_t = w_t + w_{t-1}$, i.e. $u_t$ has a one-step positive correlation and the pole is at $-1$. In general, for a sequence of random variables to converge to a standard normal, it is necessary for their variance to be unbounded. Consider a variant of the above system
\[ 
y_t = - \rho y_{t-1} + u_t,
\]
where $u_t = w_t + w_{t-1}$ such that $\E{w_i^2} = 1$ for all $i$, and $\E{w_i} = 0$. We have
\begin{align*}
    Y(z) = -\rho z^{-1} Y(z) + U(z),
\end{align*}
or equivalently that 
\[ 
Y(z) = \frac{z}{z+\rho} U(z),
\]
i.e. $y_t = \sum_{i=1}^t (-\rho)^{t-i} u_i$. Here the variance is given by
\begin{align*}
    \sigma_t^2 &= \!\mathbb{E}\Bigg[ \Bigg(\sum_{i=1}^t (-\rho)^{t-i} u_i \Bigg)^2 \Bigg] = \mathbb{E}\Bigg[ \Bigg(\sum_{i=1}^t (-\rho)^{i} u_{t-i} \Bigg)^2 \Bigg]
    \\
    & = \!\sum_{i=1}^t 2\rho^{2i} \!-\! \rho \sum_{i=1}^t \rho^{2i} \!-\! \frac 1\rho \sum_{i=1}^t \rho^{2i} = \pbr{2\!-\!\rho \!-\! \frac1\rho}\frac{1}{1\!-\!\rho^2},
\end{align*}
i.e. $\sigma_t \nrightarrow \infty$ as $\rho\to 1$. One approach to overcoming this would be to have correlation decay. For instance, having $u_t = w_t + a w_{t-1}$ for some $0 < a < 1$ would yield
\[ 
\sigma_t^2 = \pbr{2 - \rho a - \frac a \rho} \frac{1}{1-\rho^2},
\]
since $\E{u_t u_{t+1} } = a$ and $\E{u_t^2} = 2 $. This implies that $\sigma_t\to\infty$ as $\rho\to 1$.

\subsubsection{When the dominant pole is purely imaginary and noise is positively correlated}

Consider the following system
 \[ 
y_t = - y_{t-2} + u_t,
 \]
 where $u_t =  w_t + w_{t-2}$. Then, such a system satisfies
 \[ 
    y_k = 
    \begin{cases} 
    w_k + (-1)^{\frac k 2 + 1}, & k \mbox{ even} \\
    w_k + (-1)^{\frac{k+1}{2}}w_0, & k\mbox{ odd}
    \end{cases},
 \]
 i.e. the response $y_t$ does not converge to a standard normal. The $Z$-transform for the above system is given by $Y(z) = - z^2 Y(z) + U(z)$, or equivalently
 \[ 
Y(z) = \frac{z^2}{z^2+1} U(z),
 \]
 i.e. the poles of the system are purely imaginary.

\section{CONCLUSIONS}
A long standing conjecture that is used in the analysis of control systems is that the output of a stable linear system driven by noise asymptotically has a Gaussian distribution when the dominant pole is close to the unit circle. To the best of our knowledge there has been no proof of this fact, so in this paper we present a set of sufficient conditions under which the conjecture holds. Further, we argue that the conjecture does not hold in general by providing examples of noise sequences and linear systems for which the conjectured result does not hold. Since our conditions in this paper are only sufficient conditions, an interesting directions for future work is to understand whether these conditions can be relaxed. Another possible direction is to extend the results to the case where the driving noise is either a martingale or a Markov process by exploiting recent results \cite{srikant2024rates,rollin2018quantitative} on the application of Stein's method to these types of random processes.

\section{Acknowledgments}
Y.M. and R.S. were supported by NSF grant CCF-2207547.





\appendix
\label{appendix}
\subsection{Structural Assumptions on the Noise Sequence in \eqref{eq:finalsystem_mb}}
\label{subsection:motivation}

In analyzing \eqref{eq:finalsystem_mb}, we made the following assumption:
\( \{u_i\} \) be a sequence of random variables with a local dependence structure, meaning that each \( u_i \) is dependent on at most \( D \) other variables \( u_j \). This assumption is justified if \eqref{eq:finalsystem_mb} is derived from an ARMA model:
\begin{equation}
   \label{eq:dtsprocess}
   y_t = \sum_{i=1}^n a_i y_{t-i} + \sum_{j=0}^m b_j w_{t-j}
\end{equation}
where $\{w_t\}$ is a mean zero, independent noise sequence (not necessarily identically distributed). This model can be seen as an IIR filter, with the coefficients $\{a_i\}$ forming the recursive (feedback) part and the coefficients $\{b_j\}$ constituting the feedforward part, resulting in an impulse response that is theoretically infinite, yet convergent under appropriate stability conditions. Alternatively, it is interpretable as an ARMA$(n,m)$ process, where the autoregressive component captures the dependence of $y_t$ on its past values and the moving average component reflects the influence of current and past independent noise inputs, thereby providing a probabilistic framework for analyzing the system’s asymptotic behavior.

The process in \eqref{eq:dtsprocess} is equivalent to 
\begin{equation}
   \label{eq:dtsprocess1}
   y_t = \sum_{i=1}^n a_i y_{t-i} + u_t
\end{equation}
where,
\begin{equation}
   \label{eq:ut}
   u_t:= \sum_{j=0}^m b_j w_{t-j}
\end{equation}
Note that the sequence $\{u_i\}$ is independent when $u_i = b_0 w_i$. However, in general, from \eqref{eq:ut} it is clear that $u_t$ is correlated with $u_{t+1},\ldots,u_{t+m}$ and $u_{t-1},\ldots,u_{t-m}$. The local dependence neighborhood of $u_t$, denoted as  $\mathcal{N}_{u_t}$ is defined as the set of all $u_i$ that is correlated with $u_t$. That is,
\begin{equation}
   \mathcal{N}_{u_t} = \{u_{t-m},\ldots,u_{t-1},u_t,u_{t+1},\ldots,u_{t+m}\}
\end{equation}
The maximum dependency degree of the sequence $\{u_i\}$ is thus defined as:
\begin{equation}
   D:=\max_{t}|\mathcal{N}_{u_t}| = 2m+1
\end{equation}

\subsection{The System Response $G$}
\label{subsection:System_Response}
The z-transform of \eqref{eq:dtsprocess1} yields the following,
\begin{equation}
   \label{eq:Ztransform}
   Y[z] = \sum_{i=1}^n a_i z^{-i} Y[z] + U[z]
\end{equation}
where $Y[z]$ and $U[z]$ are the Z-transforms of $y_t$ and $u_t$ respectively. \eqref{eq:Ztransform} is equivalent to the following:
\begin{equation}
   \label{eq:Ztransform1}
   Y[z] = \frac{z^n U[z]}{z^n - a_1 z^{n-1}-a_2 z^{n-2}\cdots -a_n}
\end{equation}
Let $\rh_i$ be the roots of the polynomial in the denominator of \eqref{eq:Ztransform1}. Then,
\begin{equation}
   \label{eq:finalZtransform}
   Y[z] = \frac{z^n U[z]}{\Pi_{i=1}^p \left( z- \rh_i\right)^{d_i}}
\end{equation}
where $d_i$ represents the degree of root $\rh_i$ and $p$ represents the number of distinct roots. Taking the inverse Z-transform of the system in \eqref{eq:finalZtransform},
\begin{equation}
   \label{eq:finalsystem}
   y_t = \sum_{i=1}^t G_{t-i} u_i 
\end{equation}
where $G_j = \sum_{k=1}^p c_k(j) \rh_k^j$ and the initial conditions are set as follows: $w_0, \ldots, w_{1-m}, y_0, \ldots, y_{1-n} = 0$
\begin{remark}
   $c_k(j)$ is a polynomial in $j$ of degree $d_k-1$ associated with root $\rh_k$, whose coefficients are time varying. The coefficients of $c_k(j)$ are a function of all the roots $\rh_i$. When $\rh_k$ is real, the coefficients of this polynomial are not time varying. When $\rh_k$ is complex, the coefficients of this polynomial change based on the value of $j$ mod $4$.  That is, the time varying coefficients take values from a finite set.
\end{remark}
\begin{remark}
   For the system in \eqref{eq:finalsystem} to be stable we require $|\rh_k|<1 \quad \forall k\in\{1,\ldots, p\}$.
\end{remark}
Since $\mathbb{E}[w_i]=0$, from \eqref{eq:ut}, it is true that $\mathbb{E}[u_i] = 0 \quad \forall i$. This thus implies that $\mathbb{E}[y_t]=0 \quad \forall t.$ 

Recalling the definition $G_i = \sum_{k=1}^p c_k(i) \rh_k^i$, it follows that there are two cases to consider: the dominant root is either real or complex. If real, we have 
\begin{align}  \label{eq: ordered-roots}
1 > |\rh_1| > |\rh_2| \geq |\rh_3| \geq \cdots \geq |\rh_p|.
\end{align}
If complex, it must be that $|\rh_1| = |\rh_2|$, since complex roots occur as conjugates. It follows therefore that $d_1 = d_2$. For ease of notation, we combine the polynomials $c_1(j)$ and $c_2(j)$ and express the combined polynomial as $c_1(j) |r_1|^j$, where the new $c_1(j)$ is a function of the old $c_1(j)$ and $c_2(j)$. With the new notation, we have that~\eqref{eq: ordered-roots} holds true in this case as well, with the new $r_2$ being the old $\rh_3$. Thus,
\begin{equation}
   \label{eq:G_in_terms_of_r}
   G_i = c_1(i) \cdot  r_1^i \pbr{1+ \sum_{k=2}^p \frac{c_k(i)}{c_1(i)}}\pbr{\frac{r_k}{r_1}}^i,
\end{equation}
where $c_k(i)/c_1(i)$ is the ratio of two polynomials and $|r_k/r_1| <1$ for all $k\geq 2$. Consequently,
\[ 
\lim_{i\to\infty} \pbr{1+ \sum_{k=2}^p \frac{c_k(i)}{c_1(i)} \pbr{\frac{r_k}{r_1}}^i} = 1.
\]
Therefore, there exists $T_{\eps}$ such that for all $i > T_{\eps}$, 
\begin{equation}
   \label{eq:largeTbehaviour}
   e^{-\eps} \leq 1+ \sum_{k=2}^p \frac{c_k(i)}{c_1(i)} \pbr{\frac{r_k}{r_1}}^i  \leq e^{\eps}.
\end{equation}
Furthermore, when $i > T_{\eps}$, there exist constants $\co_{d_1}$ and $\cu_{d_1}$ such that
\[ 
\cu_{d_1} i^{d_1-1} \leq |c_1(i)| \leq \co_{d_1} i^{d_1-1},~~~~ \forall \ |r_1| \in (0,1].
\]
Therefore,
\begin{align*}
    |G_i| \leq e^{\eps} \cdot \co_{d_1} i^{d_1-1} |r_1|^i, ~~~~ 
    |G_i| \geq e^{-\eps} \cdot \cu_{d_1} i^{d_1-1} |r_1|^i. 
\end{align*}
Let $\alpha := - \log|r_1|$, so that $|r_1|^i = e^{-\alpha i}$. Then,
\begin{align*} 
|G_i| &\leq e^{\eps} \cdot \co_{d_1} i^{d_1-1} |r_1|^i = \co_{d_1} \cdot e^{\eps} \cdot e^{(d_1-1)\log(i) - \alpha i}, \\
|G_i| &\geq e^{-\eps} \cdot \cu_{d_1} i^{d_1-1} |r_1|^i = \cu_{d_1} \cdot e^{-\eps} \cdot e^{(d_1-1)\log(i) - \alpha i}.
\end{align*}

\subsubsection{Positively correlated $\{u_i\}$ with a real positive dominant root $r_1$}
\label{subsubsection:pos_cor}
Given the system in \eqref{eq:dtsprocess}, and the fact that $\E{w_k^2}>0 \ \forall k$, one can verify that the sequence \( \{u_i\} \) exhibits positive correlation when all coefficients \( b_j \) in \eqref{eq:dtsprocess} are strictly positive.

Since the dominant root \( r_1 \) is real, the corresponding term \( c_1(i) \) is a fixed polynomial in time \( i \), with time-invariant coefficients. If \( r_1 \) is a repeated root with multiplicity \( d_1 \), then \( c_1(i) \) is a polynomial of degree \( d_1 - 1 \). As a real polynomial of degree \( d_1 - 1 \), it can have at most \( d_1 - 1 \) real zeros. Therefore, there exists a finite time \( T' \) such that \( c_1(t) \) is either strictly positive or strictly negative for all \( t > T' \). Consequently, for sufficiently large \( i \) and \( j \),
\begin{equation}
   \label{eq:c1ic1j}
   c_1(i) c_1(j) = |c_1(i) c_1(j)| > 0.
\end{equation}
From \eqref{eq:G_in_terms_of_r} it thus follows that,
\begin{align*}
   G_i G_j &= c_1(i) \cdot c_1(j) \cdot  r_1^{i+j} \cdot \pbr{\!1\!+\! \sum_{k=2}^p \frac{c_k(i)}{c_1(i)}\!\pbr{\!\frac{r_k}{r_1}\!}^i} \\ & ~~~~~~~~~~~~~~~~~~~~~~~~\cdot \pbr{\!1\!+\!\sum_{k=2}^p \frac{c_k(j)}{c_1(j)}\!\pbr{\frac{r_k}{r_1}}^j}
\end{align*}
When \( i, j \geq T_{\epsilon'},\) where \(T_{\epsilon'}:=\max\{T_\epsilon, T'\} \), it follows from \eqref{eq:largeTbehaviour} and \eqref{eq:c1ic1j}, along with the fact that \( r_1 > 0 \), that
\begin{equation}
   G_i G_j = |G_i||G_j| > 0.
\end{equation}

\subsection{General Correlation Decay}
\label{subsection:GCD}
Let $\mathcal{N}_i := \cbr{j: u_j \sim u_i}$ denote the \textit{indices} of the neighbors of $u_i$. The variance $\sigma_t^2$ can be computed as
\begin{align}
    \sigma_t^2 &= \mathbb{E}\Bigg[\sum_{i=1}^t \sum_{j\in \mathcal{N}_{i}} G_{t-i}u_i G_{t-j}u_j\Bigg] \nonumber \\
    &= \mathbb{E}\Bigg[\sum_{i=1}^t \sum_{t-j \in \mathcal{N}_{t-i}} \!\!\! G_i G_j u_{t-i} u_{t-j}\Bigg] 
    \nonumber \\
    &= \mathbb{E}\Bigg[\sum_{i=1}^t G_i^2 u_{t-i}^2 + \sum_{i=1}^t \sum_{\substack{t-j \in \mathcal{N}_{t-i} \\ j \neq i}} G_i G_j u_{t-i}u_{t-j} \Bigg] 
    \nonumber \\
    & = c(T_{\eps'}) + \mathbb{E}\Bigg[\sum_{i=T_{\eps'}}^t G_i^2 u_{t-i}^2 + \!\!\sum_{i=T_{\eps'}}^t \!\sum_{\substack{t-j \in \mathcal{N}_{t-i} \\ j \neq i}}\!\!\!\! G_i G_j u_{t-i}u_{t-j} \Bigg] 
\end{align}
Recall that, for $i>T_{\eps'}$,
\begin{align}
    G_i = \sum_{k=1}^p c_k(i) r_k^i \!= c_1(i) r_1^i \underbrace{\pbr{1 \!+\! \sum_{k=2}^p \frac{c_k(i)}{c_1(i)} \cdot \pbr{\frac{r_k}{r_1}}^i}}_{> e^{-\eps}> 0},
\end{align}
and so 
\[ 
|G_i| \geq |c_1(i)| |r_1^i| e^{-\eps} > \cu_{d_1} i^{d_1'} |r_1^i| e^{-\eps} = \cu_{d_1} e^{-\eps} e^{-\alpha i} i^{d_1'},
\]
where $\alpha = - \log|r_1|$ is positive and $d'_1:= d_1-1$. Here $d_1$ is the multiplicity of dominant pole $r_1$. Thus,
\[ 
G_i^2 \geq |G_i|^2 \geq \cu_{d_1}^2 e^{-2\eps} e^{-2\alpha i} i^{2d'_1}.
\]
Let $\su_2>0$ and $\so_2>0$ be such that $\su_2 \leq \E{u_i^2} \leq \so_2$ for all $i$. Then,
\begin{align*}
    \mathbb{E}\Bigg[ \sum_{i=T_{\eps}}^t G_i^2 u_{t-i}^2\Bigg] \geq \sum_{i=T_{\eps}^t} \cu_{d_1}^2 e^{-2\eps} e^{-2\alpha i} i^{2d_1'} \su_2.
\end{align*}

Consider the cross terms, which satisfy $|\E{u_i u_j}| \leq a \su_2$ for some $a \in (0,1)$ and $\forall i\neq j$. Note that for all $i > T_{\eps'}$,
\[ 
|G_i| \leq \cu_{d_1} e^{\eps} i^{d_1'} |r_1|^i  = \co_{d_1} e^{\eps} e^{-\alpha i} i^{d_1'},
\]
and so it follows that for all pairs $i,j > T_{\eps'}$,
\[ 
|G_i| \cdot |G_j| \cdot |\E{u_{t-i}u_{t-j}}| \leq \co_{d_1}^2 e^{2\eps} e^{-\alpha (i+j)} (ij)^{d_1'} \cdot a\su_2
\]
Thus, since $x \geq  -|x|$ for all real $x$, it follows that
\begin{align}
 \mathbb{E}\Bigg[\sum_{i=T_{\eps'}}^t & \sum_{ \substack{t-j\in\mathcal{N}_{t-i} \\ j \neq i} } \!\!\!\! G_i G_j u_{t-i} u_{t-j} \Bigg]
 \nonumber \\
 & \geq - \sum_{i=T_{\eps'}}^t \sum_{ \substack{t-j\in\mathcal{N}_{t-i} \\ j \neq i} } \co_{d_1}^2 e^{2\eps} e^{-\alpha(i+j)} (ij)^{d_1'} a\su_2 
 \nonumber \\
 & = - \sum_{i=T_{\eps'}}^t \co_{d_1}^2 e^{2\eps} a\su_2 i^{d_1'} e^{-\alpha i} \sum_{ \substack{t-j\in\mathcal{N}_{t-i} \\ j \neq i} } e^{-\alpha j } j^{d_1'} 
 \nonumber  \\
 & \geq - \sum_{i=T_{\eps'}}^t \co_{d_1}^2 e^{2\eps} a\su_2 i^{d_1'} (i+M)^{d_1'} \cdot 2M e^{\alpha M}e^{-2\alpha i},
\end{align}
where the last step uses that $|\widetilde{\mathcal{N}}_{t-i} | \leq 2M$ and that $i-M \leq j\leq i+M$. Here $\widetilde{\mathcal{N}}_{t-i}:= {\mathcal{N}}_{t-i}\backslash u_{t-i}.$ Therefore,
\begin{align}
    \sigma_t^2 &\geq c(T_{\eps'}) + \sum_{i=T_{\eps'}}^t \cu_{d_1}^2e^{-2\eps}e^{-2\alpha i} i^{2d_1'} \su_2 
    \nonumber \\
    & ~~~~ - \sum_{i=T_{\eps'}}^t \co_{d_1}^2 e^{2\eps} a\su_2 i^{d_1'}(i+M)^{d_1'} 2 M e^{\alpha M} e^{-2\alpha i} 
    \nonumber  \\
    & = c(T_{\eps'}) + \sum_{i=T_{\eps'}}^t \su_2\bigg( \cu_{d_1}^2 e^{-2\eps} i^{2d_1'} 
    \nonumber \\
    & ~~~~ - \cu_{d_1}^2 e^{2\eps} a \cdot 2M e^{\alpha M} i^{d_1'}(i+M)^{d_1'} \bigg) e^{-2\alpha i}.
\end{align}
As $\alpha \to 0$, we have
\[ 
\sigma_t^2 \geq g(t) :=  c(T_{\eps'}) + \su_2(t-T_{\eps'})c''.
\]
Note that $g(t) \to \infty$ as $t\to\infty$ if $a$ satisfies the following:
\begin{align*}
    a \leq \frac{\cu_{d_1}^2 e^{-2\eps} i^{2d_1'}}{\co_{d_1}^2 e^{2\eps} 2M e^{\alpha M} i^{d_1'}(i+M)^{d_1'}} = \frac{C i^{d_1'}}{(i+M)^{d_1'}},
\end{align*}
i.e. $a <\delta C$ where 
\[
\delta \leq \min_{i>T_{\eps'}} \frac{i^{d_1'}}{(i+M)^{d_1'}}.
\]
Such a $\delta$ exists since $(i+M)^{d_1'}$ is monotonically increasing. Then, $\sigma_t^2 \to \infty$ as $t\to\infty$ at rate $O(t)$. Proving convergence of $d_\mathcal{W}\pbr{\frac{y_t}{\sigma_t},Z}\to 0$ is now identical to the previous cases.

\end{document}